\documentclass{llncs}

\usepackage{amsmath,amssymb,stmaryrd,upgreek,pifont,rotating}
\usepackage{graphicx}
\usepackage{color}

\renewcommand{\phi}{\varphi}

\def\squareforqed{\hbox{\rlap{$\sqcap$}$\sqcup$}}
\def\qed{\ifmmode\squareforqed\else{\unskip\nobreak\hfil
\penalty50\hskip1em\null\nobreak\hfil\squareforqed
\parfillskip=0pt\finalhyphendemerits=0\endgraf}\fi}

\newcommand{\sk}{\mathsf{skip}}
\newcommand{\ite}[3]{\mathsf{if}\  #1\  \mathsf{then}\  #2\  \mathsf {else}\  #3\ }
\newcommand{\while}[2]{\mathsf{while}\  #1\  \mathsf{do}\  #2\ }

\newcommand{\h}{\mathsf{H} }
\newcommand{\U}{\mathsf{U} }
\newcommand{\phase}{\mathsf{T} }
\newcommand{\cnot}{\mathsf{CNot} }
\newcommand{\true}{\mathsf{true} }
\newcommand{\false}{\mathsf{false} }

\newcommand{\den}[1]{\llbracket #1 \rrbracket}
\newcommand{\dena}[1]{\den{#1}^\natural}

\newcommand{\Fa}{F^\natural}

\newcommand{\prog}[2]{\langle #1,#2\rangle}

\newcommand{\id}{\mathbb I}

\newcommand{\trace}{\textup{tr}}
\newcommand{\proj}{\textup P}

\newcommand{\RR}{\Uppi}
\newcommand{\AAA}{\mathcal A}
\newcommand{\joina}{\vee}
\newcommand{\meeta}{\wedge}
\newcommand{\Joina}{\bigvee}

\newcommand{\poa}{\le}

\newcommand{\DD}{\mathcal D}

\title{Quantum Entanglement Analysis based~on Abstract~Interpretation} 
\author{Simon Perdrix}
\institute{Oxford University Computing Laboratory\\\email{simon.perdrix@comlab.ox.ac.uk}}
\begin{document}

\maketitle

\begin{abstract}
Entanglement is a non local property of quantum states which has no classical counterpart and plays a decisive role in quantum information theory. Several protocols, like the teleportation, are based on quantum entangled states. Moreover, any quantum algorithm which does not create entanglement can be efficiently simulated on a classical computer. The exact role of the entanglement is nevertheless not well understood. Since an exact analysis of entanglement evolution induces an exponential slowdown, we consider approximative analysis based on the framework of abstract interpretation. In this paper, a concrete quantum semantics based on superoperators is associated with a simple quantum programming language. The representation of entanglement, i.e. the design of the abstract domain is a key issue. A representation of entanglement as a partition of the memory is chosen. An abstract semantics is introduced, and the soundness of the approximation is proven.

\end{abstract}

\section{Introduction}

Quantum entanglement is a non local property of quantum mechanics. The entanglement reflects the ability of a quantum system composed of several subsystems, to be in a state which cannot be decomposed into the states of the subsystems. 
Entanglement is one of the properties of quantum mechanics which caused Einstein and others to dislike the theory. In 1935, Einstein, Podolsky, and Rosen formulated the EPR paradox \cite{EPR35a}. 

On the other hand, quantum mechanics has been highly successful in producing correct experimental predictions, and the strong correlations associated with the phenomenon of quantum entanglement have been observed indeed \cite{AGR81a}.

Entanglement leads to correlations between subsystems that can be
exploited in information theory (e.g., teleportation scheme
\cite{BBC+93a}). The entanglement plays also a decisive, but not yet well-understood, role in quantum computation, since any quantum algorithm can be efficiently simulated on a classical computer when the quantum memory is not entangled during all the computation. As a consequence, interesting quantum algorithms, like Shor's algorithm for factorisation \cite{shor94}, exploit this phenomenon.

In order to know what is the amount of entanglement of a quantum state, several measures of entanglement have been introduced (see for instance \cite{NC00}). Recent works consist in characterising, in the framework of the one-way quantum computation \cite{VMDB06}, the amount of entanglement necessary for a universal model of quantum computation. Notice that all these techniques consist in analysing the entanglement of a given state, starting with its mathematical description. 

In this paper, the  entanglement
\emph{evolution} during the computation is analysed. The  description of quantum
evolutions is done via a simple quantum programming language. 
The development of such quantum programming languages is recent, see \cite{SelingerP:brisqp,Gay06} for a survey on this topic.

An exact analysis of entanglement evolution induces
an exponential slowdown of the computation. 
Model checking techniques have been introduced
\cite{GRP07} including entanglement. Exponential slowdown of such analysis is avoided by reducing the domain to stabiliser states (i.e. a subset of
quantum states that can be efficiently simulated on a classical
computer). As a consequence, any quantum program that cannot be
efficiently simulated on a classical computer cannot be analysed. 

Prost and Zerrari \cite{PZ08} have recently introduced a logical entanglement analysis for functional languages. This logical framework allows analysis of higher-order functions, but does not provide any static analysis for the quantum programs without annotation. Moreover, only pure quantum states are considered.

In this paper, we introduce a novel approach of entanglement analysis 
based on the framework of abstract interpretation \cite{CC77}. A concrete quantum semantics based on superoperators is 
associated with a simple quantum programming language. The representation of 
entanglement, i.e. the design of the abstract domain is a key issue. A 
representation of entanglement as a partition of the memory is chosen. An 
abstract semantics is introduced, and the soundness of the approximation is 
proved.

\section{Basic Notions and Entanglement}

\subsection{Quantum Computing}

We briefly recall the basic definitions of quantum computing; please refer to Nielsen and Chuang \cite{NC00} for a complete introduction to the subject. 

The state of a quantum system can be described by a density matrix, i.e. a self adjoint\footnote{$M$ is self adjoint (or Hermitian) if and only if $M^\dagger =M$} positive-semidefinite\footnote{$M$ is positive-semidefinite if all the eigenvalues of $M$ are non-negative.} complex matrix of trace\footnote{The trace of $M$ ($\trace(M)$) is the sum of the diagonal elements of $M$} less than one. The set of density matrices of dimension $n$ is $D_{n} \subseteq \mathbb C^{n\times n}$.

The basic unit of information in quantum computation is a quantum bit
or \emph{qubit}. The state of a single qubit is described by a
$2\times 2$ density matrix $\rho \in D_{2}$. The state of a register
composed of $n$ qubits is a $2^n\times 2^n$ density matrix. If two 
registers $A$ and $B$ are in states $\rho_A \in D_{2^n}$ and $\rho_B
\in D_{2^m}$, the composed system $A,B$ is in state $\rho_A \otimes
\rho_B\in D_{2^{n+m}}$.

The basic operations on quantum states are unitary operations and measurements. A unitary 
operation maps an $n$-qubit state to an $n$-qubit state, and is given
by a $2^n\times 2^ n$-unitary matrix\footnote{$U$ is unitary if and
only if $U^\dagger U=UU^\dagger=\id$.}.  
If a system in state $\rho$ evolves according to a unitary
transformation $U$, the resulting density matrix is $U \rho
U^\dagger$. The parallel composition of two unitary transformations
$U_A$, $U_B$ is $U_A \otimes U_B$.

The following unitary transformations form an approximative universal family of unitary transformations, i.e. any unitary transformation can be approximated by composing the unitary transformations of the family \cite{NC00}.

$$H=\frac1{\sqrt 2}\left(\begin{array}{cc} 1&1\\ 1&-1\end{array}\right), T = \left(\begin{array}{cc} 1&0\\ 0&e^{i\pi/4}\end{array}\right),CNot =\left(\begin{array}{cccc} 1&0&0&0\\ 0&1&0&0\\0&0&0&1\\ 0&0&1&0\end{array}\right)$$

$$\sigma_{x}=\left(\begin{array}{cc} 0&1\\ 1&0\end{array}\right),\sigma_{y}=\left(\begin{array}{cc} 0&-i\\i&0\end{array}\right),\sigma_{z}=\left(\begin{array}{cc} 1&0\\ 0&-1\end{array}\right)$$

\ 

\ 

A measurement is described by a family of projectors
$\{\proj_{x},x \in X\}$ satisfying $\proj_i^2=\proj_i$, $\proj_i\proj_j=0$ if $i\neq j$,  
 and $\sum_{x\in
X}\proj_{x}=\id$. 
A computational basis measurement is $\{ \proj_k, 0\le k<
2^n\}$, where $\proj_k$ has $0$ entries everywhere except one $1$ at row
$k$, column $k$. The parallel composition of two measurements
$\{\proj_x,x \in X\}$, $\{\proj'_y, y \in Y\}$ is $\{\proj_x\otimes \proj'_y, (x,y)
\in X\times Y\}$.

According to a probabilistic interpretation, a measurement according
to $\{\proj_x,x\in X\}$ of a state $\rho$ produces the classical
outcome $x\in X$ with probability $\trace(\proj_x \rho \proj_x)$ and
transforms $\rho$ into $ \frac{1}{\trace(\proj_x \rho \proj_x)}\proj_x
\rho \proj_x$. 

Density matrices is a useful formalism for representing probability
distributions of quantum states, since the state $\rho$ of a system which is
in state $\rho_1$ (resp. $\rho_2$) with probability $p_1$
(resp. $p_2$) is $\rho= p_1 \rho_1 + p_2 \rho_2$. As a consequence, a measurement according to $\{\proj_{x},x \in X\}$
transforms $\rho$ into $\sum_{x \in X} \proj_{x} \rho \proj_{x}$.

Notice that the sequential compositions of two
measurements (or of a measurement and a unitary transformation) is no
more a measurement nor a unitary transformation, but a superoperator,
i.e. a trace-decreasing\footnote{$F$ is trace decreasing iff
$\trace(F(\rho))\le \trace(\rho)$ for any $\rho$ in the domain of $F$. Notice that
superoperators are sometimes defined as trace-perserving maps, 
however trace-decreasing is more suitable in a semantical context, see
\cite{Sel03} for details.} completely positive\footnote{$F$ is
positive if $F(\rho)$ is positive-semidefinite for any positive $\rho$ in the
domain of $F$. $F$ is completely positive if $\id_k\otimes F$ is positive
for any $k$, where $\id_k: \mathbb C^{k\times k} \to \mathbb C^{k\times
k}$ is the identity map.} linear map. Any quantum evolution can be
described by a superoperator. 

The ability to initialise any qubit in a given state $\rho_0$, to apply any
unitary transformation from a universal family, and to perform a
computational measurement are enough for  simulating any superoperator.

\subsection{Entanglement}
\label{sec:entanglement}

Quantum entanglement is a non local property which has no classical
counterpart. Intuitively, a quantum state of a system composed of
several subsystems is \emph{entangled} if it cannot be decomposed into the state
 of its subsystems. A quantum state which is not entangled is called
\emph{separable}. 

More precisely, for a given finite set of qubits $Q$, let $n=|Q|$. For a
given partition $A,B$ of $Q$, and a given $\rho \in D_{2^n}$, $\rho$
is biseparable according to $A,B$ (or $(A,B)$-separable for short) if and only if there exist $K$, $p_{k}\ge 0$, $\rho_{k}^{A}$ and $\rho_{k}^B$ such that $$\rho = \sum_{k\in K} p_{k}\rho^{A}_{k}\otimes \rho^{B}_{k}$$

$\rho$ is entangled according to the partition $A,B$ if and only if
$\rho$ is not $(A,B)$-separable. 

Notice that biseparability provides a very partial information about the
entanglement of a quantum state, for instance for a $3$-qubit state
$\rho$, which is $(\{1\},\{2,3\})$-separable, qubit $2$ and qubit $3$
may be entangled or not. 

One way to generalise the biseparability is to consider that a quantum
state is $\pi$-separable -- where $\pi=\{Q_j,j\in J\}$ is a
partition of $Q$ -- if and only if there exist $K$, $p_{k}\ge 0$, and $\rho_{k}^{Q_j}$ such that $$\rho = \sum_{k\in K}
p_{k}\left(\bigotimes_{j\in J}\rho^{Q_j}_{k}\right)$$

Notice that the structure of quantum entanglement presents some interesting and non trivial properties. For instance there exist some 3-qubit states $\rho$ such that $\rho$ is bi-separable for any bi-partition of the 3 qubits, but not fully separable i.e., separable according to the partition $\{\{1\},\{2\},\{3\}\}$. 
As a consequence, for a given quantum state, there is not necessary a \emph{best representation} of its  entanglement.

\subsection{Standard and diagonal basis}
\label{sec:basis}

For a given state $\rho \in \DD^Q$ and a given qubit $q\in Q$, if $\rho$ is $(\{q\},Q\setminus \{q\})$-separable, then $q$ is separated from the rest of the memory. Moreover, such a qubit may be a basis state in the standard basis ($\bf s$) or the diagonal basis ($\bf d$), meaning that the state of this qubit can be seen as a 'classical state' according to the corresponding basis.  

More formally, a qubit $q$ of $\rho$ is in the standard basis if there exists $p_0,p_1 \ge 0 $, and $\rho_0, \rho_1\in \DD^{Q\setminus \{q\}}$ such that $\rho=p_0P_q^\true\otimes \rho_0 + p_1P_q^\false\otimes \rho_1$. 
Equivalently, $q$ is in the standard basis if and only if $P^\true_q \rho P^\false_q= P^\false_q \rho P^\true_q = 0$. A qubit $q$ is in the diagonal basis in $\rho$ if and only if $q$ is in the standard basis in $H_q\rho H_q$. 

Notice that some states, like the maximally mixed 1-qubit state $\frac 12(P^\true+P^\false)$ are in both standard and diagonal basis, while others are neither in standard nor diagonal basis like the 1-qubit state $THP^\true HT$. 

We introduce a function $\beta:\DD^Q \to B^Q$, where $B^Q = Q\to \{{\bf s}, {\bf d}, \top, \bot\}$, such that $\beta(\rho)$ describes which qubits of $\rho$ are in the standard or diagonal basis:

\begin{definition}
For any finite $Q$, let $\beta : \DD^Q \to B^Q$ such that for any $\rho \in  \DD^Q$, and any $q\in Q$,
{$$\beta(\rho)_q=\begin{cases}
\bot &\text{if $q$ is in both standard and diagonal basis in $\rho$}\\
{\bf s} &\text{if $q$ is in the standard and not in the diagonal basis in $\rho$}\\
{\bf d} &\text{if $q$ is in the diagonal and not in the standard basis in $\rho$}\\
\top& \text{otherwise}
\end{cases}$$}
\end{definition}

\section{A Quantum Programming Language}

Several quantum programming languages have been introduced
recently. For a complete overview see \cite{Gay06}. We use an imperative
quantum programming language introduced in \cite{Per07d}, the syntax is
similar to the language introduced by Abramsky \cite{Abramsky-SQPL}. For the
sake of simplicity and in
order to focus on entanglement analysis, the memory is supposed to be fixed and finite. Moreover, the memory is
supposed to be composed of qubits only, whereas hybrid memories
composed of classical and quantum parts are often
considered. However, contrary to the quantum circuit or quantum Turing
machine frameworks, the absence of classical memory does not avoid the classical control of
the quantum computation since classically-controlled conditional structures are allowed (see
section \ref{sec:concrete}.)
 
\begin{definition}[Syntax]
For a given finite set of symbols $q\in Q$, a program is a pair 
$\langle C,Q\rangle$ where $C$ is a command defined as follows:

$$\begin{array}{rccl}
C&::=&& \sk\\
&&|& C_{1}; C_{2}\\
&&|& \ite{q}{C_{1}}{C_{2}}\\
&&|& \while{q}{C}\\
&&|& \h (q)\\
&&|& \phase (q)\\
&&|& \cnot(q,q)
\end{array}$$
\end{definition}

\begin{example}\label{ex:telep} Quantum entanglement between two qubits $q_{2}$ and
$q_{3}$ can be created for instance by applying $H$ and $CNot$ on an
appropriate state. Such an entangled state can then be used to
\emph{teleporte} the state of a third qubit $q_{1}$. The protocol of
teleportation \cite{BBC+93a} can be described as $\langle
\mathsf{teleportation}, \{q_1,q_2,q_3\}\rangle$, where

$$\begin{array}{rl}
\mathsf{teleportation:}&\h(q_{2}) ;\\
& \cnot(q_{2},q_{3}) ;\\
&\cnot(q_{1},q_{2}) ;\\
& \h (q_1) ;\\
&\mathsf{if} {~q_{1}~} \mathsf{then}\\
&~~~\ite{q_{2}}{\sk}{\sigma_{x}(q_{3})}\\
&\mathsf{else}\\
&~~~\ite{q_{2}}{\sigma_{z}(q_{3})}{\sigma_{y}(q_{3})}\\
\end{array}$$

The semantics of this program is given in example \ref{ex:csem}.
\end{example}

\subsection{Concrete Semantics}
\label{sec:concrete}

Several domains for quantum computation have been introduced
\cite{K03,Abramsky-SQPL,Phd06}. Among them, the domain of
superoperators over density matrices, introduced by Selinger
\cite{Sel03} turns out to be one of the most adapted to quantum semantics. Thus, we introduce a denotational semantics following the work of Selinger.

For a finite set of variables $Q=\{q_{0}, \ldots , q_{n}\}$, let $\DD^Q= D_{2^{|Q|}}$. $Q$ is a set of qubits, the state of $Q$ is a density operator in $\DD^Q$. 
\begin{definition}[L\"owner partial order]
For matrices $M$ and $N$ in $\mathbb C ^{n\times n}$, $M\sqsubseteq N$
if $N-M$ is positive-semidefinite.

\end{definition}

In \cite{Sel03}, Selinger proved that the poset $(\DD^Q,\sqsubseteq)$ is a complete partial order with $0$ as its
least element. Moreover the poset of superoperators over $\DD^Q$ is a
complete partial order as well, with $0$ as least element and where
the partial order $\sqsubseteq'$ is defined as $F \sqsubseteq' G \iff
\forall k \ge 0, \forall \rho \in \DD_{k2^{|Q|}}, (\id_k \otimes F)(\rho)
\sqsubseteq (\id_k \otimes G)(\rho)$, where $\id_k:\DD_k \to \DD_k$
is the identity map. Notice that these complete partial orders are not
lattices (see \cite{Sel03}.)

We are now ready to introduce the concrete denotational semantics
which associates with any program $\prog C Q$, a superoperator $\den C :  \DD^Q \to \DD^Q$.

\begin{definition}[Denotational semantics]\label{def:concretesem}
$$\begin{array}{rcl}
\den{\sk}&=&\id\\
&&\\
\den{C_{1};C_{2}} &=& \den{C_{2}}\circ\den{C_{1}}\\
&&\\
\den{\U (q)} &=& \lambda \rho.U_{q} \rho U_{q}^\dagger\\
&&\\
\den{\cnot (q_{1},q_{2})} &=& \lambda \rho .CNot_{q_{1},q_{2}} \rho CNot_{q_{1},q_{2}}^\dagger\\
&&\\
\den{\ite q {C_{1}}{C_{2}}} & =&\lambda \rho. \left(\den{C_{1}} (\proj^{\true}_{q}\rho \proj^{\true}_{q}) + \den{C_{2}}(\proj^{\false}_{q}\rho \proj^{\false}_{q})\right)\\
&&\\
\den{\while q C}&=& \mathsf{lfp}\left(\lambda f. \lambda \rho. \left(f\circ \den C (\proj^{\true}_{q}\rho \proj^{\true}_{q}) + \proj^{\false}_{q}\rho \proj^{\false}_{q}\right)\right)\\
&=&\sum_{n\in \mathbb N}\left(F_{\proj^{\false}}\circ (\den C \circ F_{\proj^\true})^n\right)\\

\end{array}$$

where $\proj^{\true} = \left(\begin{array}{cc} 1&0\\ 0&0\end{array}\right)$ and $\proj^{\false} = \left(\begin{array}{cc} 0&0\\ 0&1\end{array}\right)$, $F_{M} = \lambda \rho. M\rho M^\dagger$, and $M_{q}$ means that $M$ is applied on qubit $q$. We refer 
the reader to an extended version of this paper for the technical explanations on 
continuity and convergence.
\end{definition}

In the absence of classical memory, the classical
control is encoded into the conditional structure $\ite q {C_1} {C_2}$
such that the qubit $q$ is first measured according to the computational
basis. If the first projector is applied, then the classical outcome
is interpreted as \emph{true} and the command $C_1$ is
applied. Otherwise, the second projector is applied, and the command $C_2$ is performed. The classical control
appears in the loop $\while q C$ as well.

As a consequence of the classical control, non unitary transformations can be implemented: 

$$\begin{array}{rcccl}
\den{\ite q q  {\sigma_x(q)}}&:& \DD^{\{q\}} \to \DD^{\{q\}} &=&
\lambda \rho . \proj^\true\\
&&\\
\den{\while q {\h(q)}}&:& \DD^{\{q\}} \to \DD^{\{q\}} &=& \lambda \rho . \proj^\false\\
\end{array}$$

Notice that the matrices $\proj^\true$ and $\proj^\false$, used in definition
\ref{def:concretesem} for describing the computational measurement
$\{\proj^\true,\proj^\false\}$ can also be used as density matrices for 
describing a quantum state as above.

Moreover, notice that all the ingredients for approximating any
superoperators can be encoded into the language: the ability to
initialise any qubit in a given state (for instance $\proj^\true$ or
$\proj^\false$); an approximative universal family of unitary
transformation $\{H, T, CNot, \sigma_x,\sigma_y, \sigma_z\}$; and the
computational measurement of a qubit $q$ with $\ite q \sk \sk$.

\begin{example}\label{ex:csem} The program $\prog {\mathsf {teleportation}} {\{q_1,q_2,q_3\}}$ described in example \ref{ex:telep} realises the
teleportation from $q_1$ to $q_3$, when the qubits $q_2$
and $q_3$ are  both initialised in state $\proj^\true$: for any $\rho \in \DD_2$, $$\den {\mathsf{teleportation}}(\rho \otimes \proj^\true \otimes \proj^\true)
=\left(\frac 14\sum_{k,l\in \{\true, \false\}}{\proj^k \otimes \proj^l}\right)\otimes \rho$$

\end{example}

\section{Entanglement Analysis}

What is the role of the entanglement in quantum information theory?
How does the entanglement evolve during a quantum computation? 
We consider the problem of analysing the
entanglement evolution on a classical computer, since no large scale
quantum computer is available at the moment. Entanglement analysis
using a quantum computer is left to further
investigations\footnote{Notice that this is not clear that the use of
a quantum computer avoids the use of the classical computer since there
is no way to measure the entanglement of a quantum state without
transforming the state.}.

In the absence of quantum computer, an obvious solution consists in
simulating the quantum computation on a classical computer.  
Unfortunately, the classical memory required for the
simulation is exponentially large in the size of the quantum memory of the program
simulated. Moreover, the problem \textbf{SEP} of deciding whether a given quantum
state $\rho$ is biseparable or not is NP Hard\footnote{For pure  quantum states
(i.e. $\trace (\rho^2)=\trace(\rho)$), a linear algorithm have been introduced
\cite{MJ03} to solve the sub-problem of finding biseparability of the
form $(\{q_0,\ldots, q_k\},\{q_{k+1},\ldots,
q_n\})$ -- thus sensitive to the ordering of the qubits in the
register. Notice that this algorithm is linear in the size
of the input which is a density matrix, thus the algorithm is
exponential in the number of qubits.}
\cite{Gur03b}. Furthermore, the input of the problem \textbf{SEP} is a
density matrix, which size is exponential in the number of qubits.   
As a consequence, the solution of a classical simulation is
not suitable for an efficient entanglement analysis.

To tackle this problem, a solution consists in reducing the size of the
quantum state space by considering a subspace of possible states, such
that there exist algorithms to decide whether a state of the subspace is entangled or not in a polynomial time in
the number of qubits. This solution has been developed in
\cite{GRP07}, by considering stabiliser states only. However, this
solution, which may be suitable for some quantum protocols, is questionable
for analysing quantum algorithms since all the quantum
programs on which such an entanglement analysis can be driven
are also efficiently simulable on a classical computer. 

In this paper, we introduce a novel approach which consists in
approximating the entanglement evolution of the quantum memory. This solution is based on the framework of abstract interpretation
introduced by Cousot and Cousot \cite{CC77}. Since a classical domain for driving a sound
 and complete analysis of entanglement is exponentially large in the
number $n$ of qubits, we consider an abstract domain of size
$n$ and we introduce an abstract semantics which leads to a sound
approximation of the entanglement evolution during the computation.

\subsection{Abstract semantics}

The entanglement of a quantum state can be represented as a partition
of the qubits of the state (see section \ref{sec:entanglement}), thus a natural abstract domain is a domain
composed of partitions. Moreover, for a given state $\rho$, one can add a flag for each qubit $q$, indicating whether the state of this qubit is in the standard basis $\bf s$ or in the diagonal basis ${\bf d}$ (see section \ref{sec:basis}). 

\begin{definition}[Abstract Domain] 
For a finite set of variables $Q$, 
 let $\mathcal A^Q = B^Q \times \RR^Q$ be an abstract domain, where  
 $B^Q = Q\to \{{\bf s}, {\bf d}, \top, \bot\}$ and 
$\RR^Q$ is the set of partitions of $Q$: 
 $$\RR^Q =\{\pi \subseteq \wp(Q)\setminus \{\emptyset\}~|~\bigcup_{X\in \pi} X = Q ~and~ (\forall X,Y \in \pi, ~X\cap Y
=\emptyset ~or~ X=Y)\}$$
\end{definition}

 The abstract domain $\mathcal A$ is ordered as follows. First, let $(\{{\bf s}, {\bf d}, \top, \bot\},\le)$ be a poset, where $\le$ is defined as: $\bot \le {\bf s} \le \top$ and $\bot \le {\bf d}\le \top$. $(B^Q,\le)$ is a poset, where $\le$ is defined pointwise. Moreover, for any 
  $\pi_1,\pi_2 \in \RR^Q$, 
 let $\pi_1 \le \pi_2$ if $\pi_1$ rafines $\pi_2$, i.e. for every block
$X\in \pi_1$ there exists a block $Y\in \pi_2$ such that $X\subseteq
Y$. Finally, for any $(b_1,\pi),(b_2,\pi_2)\in \mathcal A^Q$, $(b_1,\pi)\poa (b_2,\pi_2)$ if $b_1\le b_2$ and $\pi_1\le\pi_2$.
 
 \begin{proposition} For any finite set $Q$, $(\mathcal A^Q , \poa)$ is a complete partial order, with
$\bot= (\lambda q. \bot,\{ \{q\}, q\in Q\})$ as least element. 
\end{proposition}

\begin{proof}
Every chain has a supremum since $Q$ is finite.\qed
\end{proof}

Basic operations of meet and join are defined on $\mathcal A^Q$. It turns out that contrary to $\DD^Q$, $\langle \mathcal A^Q,\joina, \meeta,
\bot, (\lambda q.\top,\{Q\})\rangle $ is a lattice. 

A removal operation on partitions is introduced as follows: for a given partition $\pi=\{Q_i, i\in I\}$, let $\pi\setminus q = \{Q_i\setminus \{q\}, i \in I\} \cup \{\{q\}\}$.  
Moreover, 
for any pair of qubits
$q_1,q_2\in Q$, let $[q_1,q_2] = \{\{q ~|~ q\in Q\setminus \{q_1,q_2\}\}, \{q_1,q_2\}\}$. 

Finally, for any $b\in B^Q$, any $q_0,q \in Q$, any $k\in \{{\bf s}, {\bf d}, \top, \bot\}$, let $$b^{q_0\mapsto k}_{q} = \begin{cases} k &\text{if $q=q_0$}\\ b_q&\text{otherwise}\end{cases}$$

We are now ready to define the abstract semantics of the language:

\begin{definition}[Denotational abstract semantics]
For any program $\prog C Q$, let $\dena C :  \AAA^Q  \to \AAA^Q  $ be defined as follows:
For any $(b,\pi)\in \AAA^Q$, 
$$\begin{array}{rcl}
\dena{\sk}(b,\pi)&=&(b,\pi)\\
&&\\
\dena{C_{1};C_{2}}(b,\pi)&=& \dena{C_{2}}\circ\dena{C_{1}} (b,\pi)\\
&&\\
\dena{\sigma (q)} (b,\pi)&=&(b,\pi)\\
&&\\
\dena{\h (q)} (b,\pi)&=&(b^{q\mapsto {\bf d}},\pi) \text{ if $b_q={\bf s}$}\\
&=&(b^{q\mapsto {\bf s}},\pi) \text{ if $b_q={\bf d}$}\\
&=&(b,\pi) \text{ otherwise}\\
&&\\
\dena{\phase (q)} (b,\pi)&=&(b^{q\mapsto {\bf \top}},\pi) \text{ if $b_q={\bf d}$}\\
&=&(b^{q\mapsto {\bf s}},\pi) \text{ if $b_q={\bot}$}\\
&=&(b,\pi) \text{ otherwise}\\
&&\\

\dena{\cnot (q_{1},q_{2})}(b,\pi) &=& (b,\pi) ~~ \text{ if $b_{q_1}={\bf s}$ or $b_{q_2}={\bf d}$}\\
  &=& (b^{q_1 \mapsto {\bf s}},\pi) ~~ \text{ if $b_{q_1}=\bot$ and $b_{q_2} > {\bot}$}\\
    &=& (b^{q_2 \mapsto {\bf d}},\pi) ~~ \text{ if $b_{q_1}>\bot$ and $b_{q_2} = {\bot}$}\\
     &=& (b^{q_1 \mapsto {\bf s}, q_2 \mapsto {\bf d}},\pi) ~~ \text{ if $b_{q_1}=\bot$ and $b_{q_2}= {\bot}$}\\
    &=& (b^{q_1,q_2\mapsto {\bf \top}},\pi\joina [q_{1},q_{2}]) ~~ \text{ otherwise}\\
&&\\
\dena{\ite q {C_{1}}{C_{2}}}(b,\pi) & =& \left(\dena{C_{1}}(b^{q\mapsto {\bf s}},\pi\setminus q) 
\joina \dena{C_{2}}(b^{q\mapsto {\bf s}},\pi\setminus q)\right)\\
&&\\
\dena{\while q C}(b,\pi)&=& \mathsf{lfp}\left(\lambda f.\lambda \pi.   \left(f\circ \dena{C}(b^{q\mapsto {\bf s}},\pi\setminus q)
\joina (b^{q\mapsto {\bf s}},\pi\setminus q)\right)\right)\\

&= & {\Joina_{n\in \mathbb N}} \left (\Fa_{q}\circ ({\dena C \circ \Fa_{q}})^n\right)\\
\end{array}$$

where $\Fa_{q} = \lambda (b,\pi). (b^{q\mapsto {\bf s}},\pi\setminus q)$.
\end{definition}

Intuitively, quantum operations act on entanglement as follows: 
\begin{itemize}
\item A
$1$-qubit measurement makes the measured qubit separable from the rest
of the memory. Moreover, the state of the measured qubit is in the standard basis.
\item A $1$-qubit unitary transformation does not modify
entanglement. Any Pauli operator $\sigma\in \{\sigma_x,\sigma_y,\sigma_z\}$ preserves the standard and the diagonal basis of the qubits. Hadamard $H$ transforms a state of the standard basis into a state of the diagonal basis and vice-versa. Finally the phase $T$ preserves the standard basis but not the diagonal basis.
\item The $2$-qubit unitary transformation $CNot$, applied on
$q_1$ and $q_2$ may create entanglement between the qubits or not. It turns out that if $q_1$ is in the standard basis, or $q_2$ is in the diagonal basis, then no entanglement is created and the basis of $q_1$ and $q_2$ are preserved. Otherwise, since a sound approximation is desired, $CNot$ is abstracted into an operation which creates entanglement.
\end{itemize}

\begin{remark}
Notice that the space needed to store a partition of $n$ elements is
$O(n)$. 
Moreover, meet, join and removal and can be done in either constant or linear
time.  
\end{remark}

\begin{example} 
\label{ex:asem}
The abstract semantics of the teleportation (see example \ref{ex:telep}) is
$\dena {\mathsf{teleportation}}: \AAA^{\{q_1,q_2,q_3\}} \to
\AAA^{\{q_1,q_2,q_3\}} = \lambda (b,\pi). (b^{q_1,q_2\mapsto {\bf s}, q_3\mapsto \top},\bot)$. Thus, for any 3-qubit
state, the state of the memory after the teleportation is fully
separable.

Assume that a fourth qubit $q_4$ is entangled with $q_1$ before
the teleportation, whereas  $q_2$ and $q_3$ are in the state
$\proj^\true$. So that, the state of the memory before the teleportation is $[q_1,q_4]$-separable. 
 The abstract semantics of $\prog {\mathsf{teleportation}} {{\{q_1,q_2,q_3,q_4\}}}$ is such
that 
$$ \dena {\mathsf{teleportation}}(b,[q_1,q_4]) =(b^{q_1,q_2\mapsto {\bf s}, q_3\mapsto \top},[q_3,q_4])$$

Thus the abstract semantics predicts that $q_3$ is entangled with
$q_4$ at the end of the teleportation, even if $q_3$ never interacts
with $q_4$.
\end{example}

\begin{example}\label{ex:trap}
Consider the program $\prog {\mathsf {trap}} {\{q_1,q_2\}}$, where $$\mathsf {trap} = \cnot(q_1,q_2);\cnot(q_1,q_2)$$

Since $CNot$ is self-inverse, $\den{\mathsf{trap}}: \DD^{\{q_1,q_2\}}\to \DD^{\{q_1,q_2\}}=\lambda \rho.\rho$. For instance, $\den{\mathsf{trap}}(\frac12(P^\true+P^\false)\otimes P^\true) =\frac12(P^\true+P^\false)\otimes P^\true$.

However, if $b_{q_1}={\bf d}$ and $b_{q_2}={\bf s}$ then $$\dena{\mathsf{trap}}(b, \{\{q_1\},\{q_2\}\}) = (b^{q_1\mapsto\top, q_1\mapsto \top}, \{\{q_1,q_2\}\})$$

Thus, according to the abstract semantics, at the end of the computation, $q_1$ and $q_2$ are entangled.

\end{example}

\subsection{Soundness}

Example \ref{ex:trap} points out that the abstract semantics is an
approximation, so it may differ from the entanglement
evolution of the concrete semantics. However, in this section, we prove the soundness of the
abstract interpretation (theorem \ref{def:sound}). 

First, we define a function $\beta:\DD^Q \to B^Q$ such that $\beta(\rho)$ describes which qubits of $\rho$ are in the standard or diagonal basis:

\begin{definition}
For any finite $Q$, let $\beta : \DD^Q \to B^Q$ such that for any $\rho \in  \DD^Q$, and any $q\in Q$,
{
$$\beta(\rho)_q=\begin{cases}
{\bf s} &\text{if {\small $P^\true_q \rho P^\false_q= P^\false_q \rho P^\true_q = 0$}}\\
{\bf d} &\text{if {\small $(P^\true_q + P^\false_q) \rho (P^\true_q - P^\false_q)= (P^\true_q - P^\false_q) \rho (P^\true_q + P^\false_q) = 0$}}\\
\top& \text{otherwise}
\end{cases}$$}
\end{definition}

A natural soundness relation is then:

\begin{definition}[Soundness relation]
For any finite set $Q$, let $\sigma\in \wp(\DD^Q,\AAA^Q)$ be the
soundness relation:
$$\sigma=\{ (\rho,(b,\pi)) ~|~ \text{$\rho$ is $\pi$-separable and $\beta(\rho)\le b$}\}$$
\end{definition}

The approximation relation is nothing but the partial order $\poa$: $(b,\pi)$ is a more
precise approximation than $(b',\pi')$ if $(b,\pi)\poa (b',\pi')$. Notice that the abstract soundness assumption is satisfied: if $\rho$ is $\pi$-separable
and $\pi \poa \pi'$ then $\rho$ is $\pi'$-separable. So, $(\rho,a)\in \sigma$ and $(\rho,a)\le (\rho',a')$ imply $(\rho',a')\in \sigma$.

However, the best approximation is not ensured. Indeed, there exist
 some $3$-qubit states \cite{EW01,BDM+99a} which are separable according to
any of the $3$ bipartitions of their qubits $\{a,b,c\}$ but which are
not $\{\{a\},\{b\},\{c\}\}$-separable. Thus, the best approximation
does not exist.

However, the soundness relation $\sigma$ satisfies the following lemma:

\begin{lemma}
For any finite set $Q$, any $\rho_1,\rho_2 \in \DD^Q$, and any $a_1,a_2 \in \AAA^Q$,
$$(\rho_1,a_1), (\rho_2,a_2)\in \sigma \implies (\rho_1+\rho_2,\pi_1\joina \pi_2)\in \sigma$$
\end{lemma}

Moreover, the abstract semantics is monotonic according to the approximation relation:

\begin{lemma}
For any command $C$, $\dena C$ is $\poa$-monotonic: for any $\pi_{1}, \pi_{2} \in \AAA^Q$, $$\pi_{1}\poa \pi_{2}\implies \dena C(\pi_{1})\poa \dena C (\pi_{2})$$
\end{lemma}

\begin{proof}
The proof is by induction on  $C$.
\end{proof}

\begin{theorem}[Soundness]\label{def:sound}
For any program  $\prog C Q$, any $\rho \in \DD^Q$, and any $a\in \AAA^Q$, 
$$(\rho,a)\in \sigma \implies (\den C (\rho), \dena C (a))\in \sigma$$

\end{theorem}

\begin{proof} The proof is by induction on $C$.

\end{proof}

In other words, if $\rho$ is $\pi$-separable and $\beta(\rho)\le b$, then $\den C (\rho)$ is $\pi'$-separable and $\beta(\den C (\rho))\le b'$, where $(b',\pi')=\dena C (b,\pi)$.

\section{Conclusion and Perspectives}
In this paper, we have introduced the first quantum entanglement analysis based on abstract
interpretation. 
Since a classical domain for driving a sound and complete analysis of entanglement is exponentially large in the number of qubits, an abstract domain based on partitions has been  introduced.
Moreover, since the concrete 
domain of superoperators is not a lattice, no Galois connection can be established between concrete and abstract domains. However, despite the absence of best abstraction, the soundness of the 
entanglement analysis has been proved.

The abstract domain is not only composed of partitions of the memory, but also of  descriptions of the qubits which are in a basis state according to the standard or diagonal basis. Thanks to this additional information, the entanglement analysis is more subtle  than an analysis of interactions:  the $CNot$ transformation is not an entangling operation if the first qubit is in the standard basis or if the second qubit is in the diagonal basis. 

 A perspective, in order to reach a more precise entanglement analysis, is to introduce a more concrete abstract domain, adding for instance a third basis, since it is known that there are three mutually unbiased basis for each qubit.

A simple quantum imperative language is considered in this paper. This
language is expressive enough to encode any quantum
evolution. However, a perspective is to develop such abstract
interpretation in a more general setting 
allowing high-order functions, representation of
classical variables, or unbounded quantum memory. 
The objective is also to provide a
practical tool for analysing entanglement evolution of more
sophisticated programs, like Shor's algorithm for factorisation 
\cite{shor94}.

Another perspective is to consider that a quantum computer is
available for driving the entanglement analysis. Notice that such an analysis of entanglement evolution is not trivial,
even if a quantum computer is available, since a tomography \cite{White:07} is required to
know the entanglement of the quantum memory state\footnote{It mainly
means that in order to obtain an {approximation} of the quantum memory
entanglement, several copies of the memory state are consumed.}.

\section{Acknowledgements}
The author would like to thank Philippe Jorrand and Fr\'ed\'eric Prost for fruitful
discussions. 
The author is supported  by EC STREP FP6-033763 Foundational Structures for Quantum Information and Computation (QICS).

\end{document}